\documentclass[11pt]{article}
\usepackage{times}
\usepackage{url}
\usepackage{graphicx}
\usepackage{subfigure} 
\usepackage{authblk}
\usepackage{algorithm}
\usepackage{algorithmic}
\usepackage{amsmath, amsthm, amssymb}

\input{Qcircuit}
\newtheorem{theorem}{Theorem}[section]\def\TH{\begin{theo}}\def\HT{\end{theo}}
\newtheorem{prop}[theorem]{Proposition}\def\PRO{\begin{prop}}\def\ORP{\end{prop}}
\newtheorem{coro}[theorem]{Corollary}\def\COR{\begin{coro}}\def\ROC{\end{coro}}
\newtheorem{defi}[theorem]{Definition}\def\DE{\begin{defi}}\def\ED{\end{defi}}
\def\LE{\begin{lemme}}\def\EL{\end{lemme}}

\DeclareMathOperator*{\argmin}{arg\,min}

\newcommand{\pbfrac}[2]{\mbox{$\mbox{}^{#1}\!/_{#2}$}}

\renewcommand{\epsilon}{\varepsilon}

\title{An optimal quantum algorithm to approximate the mean and its application for approximating the median of a set of points over an arbitrary distance}

\author[$^\ast$]{Gilles Brassard}
\author[$^\dagger$]{Fr\'ed\'eric Dupuis}
\author[$^\ddagger$]{S\'ebastien Gambs}
\author[$^\ast$]{Alain Tapp}
\affil[$^\ast$]{{\it\normalsize D\'epartement d'informatique et de recherche op\'erationnelle, Universit\'e de Montr\'eal} \authorcr
{\vspace{-1mm}\it\normalsize C.P.\ 6128, Succursale Centre-Ville, Montr\'eal (QC), H3C\,3J7~~Canada}\vspace{1mm}}
\affil[$^\dagger$]{{\it\normalsize Institut f\"ur Theoretische Physik, ETH Z\"urich}\authorcr
{\vspace{-1mm}\it\normalsize Wolfgang-Pauli-Stra{\ss}e~27, 8093~Z\"urich, Switzerland}\vspace{1mm}}
\affil[$^\ddagger$]{{\it\normalsize IRISA, Campus de Beaulieu, Universit\'e de Rennes~1}\authorcr
{\vspace{-1mm}\it\normalsize Avenue du G\'en\'eral Leclerc, 35042~Rennes Cedex, France}}
\date{25 May 2011}

\begin{document}

\maketitle

\begin{abstract}
We describe two quantum algorithms to approximate the mean value of a black-box function. The first algorithm is novel and asymptotically optimal while the second is a variation on an earlier algorithm due to Aharonov. Both algorithms have their own strengths and caveats and may be relevant in different contexts. We~then propose a new algorithm for approximating the median of a set of points over an arbitrary distance function.

\vspace{1ex}
\end{abstract}

\section{Introduction}
\label{intro}

Let $F:\{0,\dots,N-1\} \rightarrow [0,1]$ be a function and $m=\frac{1}{N} \sum_{i=1}^N F(i)$ be its mean.
When $F$ is given as a black box (i.e.~an oracle), the complexity of computing the mean can be measured by counting the number of queries made to this black box. The first quantum algorithm to approximate the mean was given by Grover, whose output of the estimate $\tilde m$ was such that $|m-\tilde{m}| \leqslant \epsilon$ after $O( \frac{1}{\epsilon} \log\log( \frac{1}{\epsilon}))$ queries to the black box~\cite{grover}. Later, Nayak and Wu~\cite{nayakwu} have proven that to get such a precision, $\Omega(1/\epsilon)$ calls to $F$ are necessary, which still left a gap between the lower and upper bounds for this problem.
In~this paper, we close this gap by presenting an asymptotically optimal algorithm to approximate the mean.
We~also describe a second algorithm that is a variation of Aharonov's algorithm~\cite{aharonov}, which may be more suitable than the first one in some contexts.  
  
Afterwards, these two algorithms to approximate the mean are used in combination with the quantum algorithm for finding the minimum of D\"{u}rr and H{\o}yer~\cite{findingminimumGrover} to obtain a quantum algorithm for approximating the \emph{median} among a set of points with arbitrary black-box distance function between these points. The median, which is defined as the point with minimum average (or total) distance to the other points, can be thought of as the point that is the most representative of all the other points. \mbox{Note} that this is very different from the simpler problem of finding the median of a set of values, which has already been solved by Nayak and Wu~\cite{nayakwu}. 
Our median-finding algorithm
combines the amplitude estimation technique of Brassard, H{\o}yer, Mosca and Tapp~\cite{brassard:amplitudeamplification} with the minimum-finding algorithm of D\"{u}rr and H{\o}yer~\cite{findingminimumGrover}.

The outline of the paper is as follows. In Section~\ref{prelimin}, we present all the tools that we need, including Grover's algorithm, the quantum algorithm for computing the mini\-mum of a function and the amplitude estimation technique. In~Section~\ref{approximate_mean}, we describe two efficient algorithms to approximate the mean value of a function, which we use in Section~\ref{approximate_median} to develop our novel quantum algorithm for approximating the median of an ensemble of points for which distances between them are given by a black box. Finally, we conclude in Section~\ref{conclu} with open questions for future work.

\section{Preliminaries}
\label{prelimin}

In this section, we briefly review the quantum information processing notions that are relevant for understanding our algorithms.  A~detailed account of the field can be found in the book of Nielsen and Chuang~\cite{qipbook}. 

As is often the case in the analysis of quantum algorithms, 
we shall assume that the input to the algorithms is given in the form of a black box (or~``oracle'') that can be accessed in quantum superposition. 
In practice, the quantum black box will be implemented as a quantum circuit that can have classical inputs and outputs. 
We~shall count as our main resource the number of calls (also called ``evaluations'') that are required to that black box.

\pagebreak

\begin{theorem}[Search~\cite{groveralgorithm,BBHT}]
There exists a quantum algorithm that takes an arbi\-trary function $F:\{0,\ldots,N-1\} \rightarrow \{0,1\}$ as input and finds some $x$ such that $F(x)=1$ if one \mbox{exists} or outputs ``void'' otherwise.
Any such $x$ is called a~``\mbox{solution}''.
The~algorithm requires $O(\sqrt{N}\,)$ evaluations of $F$ if there are no solutions.
If~there are \mbox{$s>0$} solutions, the algorithm finds one with probability at least $\pbfrac23$
after $O(\sqrt{N/s}\,)$ expected evaluations of~$F$.
This is true even if the value of~$s$ is not known ahead of time.
\end{theorem}

Following Grover's seminal work, 
D\"{u}rr and H{\o}yer~\cite{findingminimumGrover} have proposed a quantum algorithm that can find the minimum
of a function with a quadratic speed-up compared to the best possible classical algorithm.

\begin{theorem}[Minimum Finding~\cite{findingminimumGrover,quantumgraphproblems}] \label{thm:minimum}
There exists a quantum algorithm $\textsf{minimum}$ that takes an arbitrary function $F:\{0,\ldots,N-1\} \rightarrow Y$ as input (for an arbitrary totally ordered range~$Y$)
and returns a pair $(i,F(i))$ such that $F(i)$ is the minimum value taken by~$F$. The algorithm finds a correct answer with probability at least $\pbfrac34$ after $O(\sqrt{N}\,)$ evaluations of $F$.
\end{theorem}

Another extension of Grover's algorithm makes it possible to approximately \emph{count} the number of solutions to a search problem~\cite{brassard:counting}. 
It~was subsequently formulated as follows.
\begin{theorem}[Counting~\cite{brassard:amplitudeamplification}] \label{thm:counting}
There exists a quantum algorithm \textsf{count} that takes an arbitrary function $F:\{0,\ldots,N-1\} \rightarrow \{0,1\}$ as input as well as some positive integer~$t$.
If~there are $s$ values of $x$ such that $F(x)=1$,
algorithm \textsf{count}$(F, t)$ outputs an integer estimate $\tilde{s}$ for $s$ such that
\[|s-\tilde{s}| < 2\pi\frac{\sqrt{s(N-s)}}{t}+\frac{\pi^2 N}{t^2}\]
with probability at least $8/\pi^2$ after exactly $t$ evaluations of~$F$. 
In~special case \mbox{$s=0$}, \textsf{count}$(F, t)$ always outputs perfect estimate~\mbox{$\tilde{s}=0$}.
\end{theorem}

The following theorem on amplitude amplification is also adapted from~\cite{brassard:amplitudeamplification}.
Its~statement is rather more technical than that of the previous theorems.
\begin{theorem}[Amplitude estimation~\cite{brassard:amplitudeamplification}] \label{AmpEst}
There exists a quantum algorithm \textsf{amplitude\_estimation} that takes as inputs two
unitary transformations $A$ and $B$ as well as some positive integer~$t$.
If
\[A\ket{0}=\alpha\ket{ \psi_0}+\beta \ket{\psi_1}\]
(where $\ket{\psi_0}$ and $\ket{\psi_1}$ are orthogonal states and $\ket{0}$ is of arbitrary dimension) and
\[B\ket{\psi_0}\ket{0}=\ket{\psi_0}\ket{0} \text{~and~}
B\ket{\psi_1}\ket{0}=\ket{\psi_1}\ket{1}\,,\]
then \textsf{amplitude\_estimation}$(A,B,t)$ outputs $\tilde{a}$, an estimate of $a=\| \beta \|^2$, such that
\[|\tilde{a}-a| \leqslant 2\pi  \frac{\sqrt{a(1-a)}}{t}+\frac{\pi^{2}}{t^{2}}\]
with probability at least $8/\pi^2$ at a cost of doing $t$ evaluations each of $A$, $A^{-1}$ and~$B$.
\end{theorem}

We shall also need the following technical result, which we derive using standard Chernoff bound arguments.
\begin{theorem}[Majority]  \label{thm:majority}
Let $B$ be a quantum black box that approximates some
function
\[F:\{0,\dots,N-1\} \rightarrow \{0,\dots,M-1\}\]
such that its output is within $\Delta$ of the true value with probability at least~$\pbfrac23$, i.e.
\[B\ket{i}\ket{0} = \sum_j \alpha_{ij} \ket{i}\ket{x_{ij}} \mbox{~~and~} \sum_{\{j:|x_{ij} - F(i)| \leqslant \Delta\}} |\alpha_{ij}|^2 \geqslant \pbfrac23\]
for all~$i$. 
Then, for all $n$ there exists a quantum black box $B_n$ that computes $F$ with its output within $2\Delta$ of the true value with probability at least $1-1/n$, i.e.
\[B_n\ket{i}\ket{0} = \sum_j \beta_{ij} \ket{i}\ket{y_{ij}} \mbox{~~and~} \sum_{\{j : |y_{ij} - F(i)| \leqslant 2\Delta\}} |\beta_{ij}|^2 \geqslant 1-1/n\]
for all~$i$.
Algorithm $B_n$ requires $O(\log n)$  calls to~$B$.
\begin{proof}
Given an input index $i$, $B_n$ calls $k$ times black box $B$ with input~$i$, where \mbox{$k = \lceil(\lg n) / D(\tfrac{3}{5}\| \tfrac{2}{3})\rceil$}
and $D(\cdot \| \cdot)$ denotes the standard Kullback-Leibler divergence~\cite{KL_divergence}
(sometimes called the relative entropy). If~there exists an interval of size $2\Delta$ that contains at least $\pbfrac35$ of the outputs, then $B_n$ outputs the midpoint of that interval. If~there is no such interval (a~very unlikely occurrence), then $B_n$ outputs~$0$.
If at least $\pbfrac35$ of the outputs are within $\Delta$ of $F(i)$, then the output of $B_n$ cannot be further than $2\Delta$ from $F(i)$ since the interval selected by $B_n$ must contain at least one of those points. By the Chernoff bound, this happens with probability at least $1 - 2^{-k D\left(\tfrac{3}{5}\left \| \tfrac{2}{3}\right.\right)} \geqslant 1 - 1/n$.
\end{proof}
\end{theorem}
Hereinafter, we shall denote by $\textsf{majority}(B, n)$ the black box $B_n$ that results from using this algorithm on black box $B$ with parameter $n$.
Note that \mbox{$D\left(\tfrac{3}{5}\left \| \tfrac{2}{3}\right.\right) > 1/100$}, hence $\textsf{majority}(B, n)$ requires less than $100 \lg n$ calls to~$B$.
Note also that the number of calls to~$B$ does not depend on~$\Delta$.

\section{Two Efficient Algorithms to Approximate the Mean}
\label{approximate_mean}

We present two different algorithms to compute the mean value of a function.
In~both algorithms, let $F:\{0,\ldots,N-1\} \rightarrow [0,1]$ be a black-box function 
and let \mbox{$m= \frac{1}{N} \sum_x F(x)$} be the mean value of $F$, which we seek to approximate. 
Without loss of generality, we assume throughout that $N$ is a power of~2.
The first algorithm assumes that $F(x)$ can be obtained with arbitrary precision at unit cost
while the second algorithm considers that the output of function $F$ is given with $\ell$ bits of precision.

\begin{algorithm}
	\caption{\textsf{mean1}$(F,N,t)$}\label{mean_algo_1}
\begin{algorithmic}\raggedright
\STATE Let 
$$\begin{array}{l}
A'\ket{x}\ket{0} = \ket{x}\big(\sqrt{1-F(x)}\ket{0}+\sqrt{F(x)}\ket{1}\big)\\[1ex]
\text{and~} A = A' \, (\textsf{H}^{\otimes \lg N} \otimes \textsf{Id})\,,
\end{array}$$
where $\textsf{H}$ is the Walsh--Hadamard transform and $\textsf{Id}$~denotes the identity transformation on 
one qubit
\STATE Let $$B\ket{x}\ket{0}\ket{0}=\ket{x}\ket{0}\ket{0}$$ and $$B\ket{x}\ket{1}\ket{0}=\ket{x}\ket{1}\ket{1}$$ 
\STATE \textbf{return} $\textsf{amplitude\_estimation}(A,B,t)$
\end{algorithmic}
\end{algorithm}

Note that in Algorithm \ref{mean_algo_1}, it is easy to implement $A'$ (and therefore~$A$ as well as~$A^{-1}$) with only two evaluations of~$F$\@.
First, $F$ is computed in a ancillary register initialized to~$\ket{0}$, then the appropriate controlled rotations are performed, and finally $F$ is computed again to reset the ancillary register back to~$\ket{0}$.
(In~practice, this transformation will be approximated to a prescribed precision.)
The following theorem formalizes the result obtained by this algorithm.

\begin{theorem}\label{thm:mean1}
Given a black-box function
$F:\{0,\ldots,N-1\} \rightarrow [0,1]$
and its mean value $m= \frac{1}{N} \sum_x F(x)$, 
algorithm \textsf{mean1} outputs $\tilde{m}$  such that $|\tilde{m} - m| \in O\left(\frac{1}{t}\right)$ 
with probability at least~$8/\pi^2$.  The algorithm requires $4t$ evaluations of $F$.
\end{theorem}

\begin{proof}
Using the same definition as in Theorem~\ref{AmpEst}, we have that 
\[ \ket{\psi_1} = \sum_x \sqrt{\frac{F(x)}{\sum_y F(y)}} \ket{x}\ket{1}
\ \ {\rm and} \ \
\beta=\sqrt{\frac{\sum_x F(x)}{N}}.\]
The algorithm \textsf{amplitude\_estimation}$(A,B,t)$ returns an estimate
\mbox{$\tilde{m}=\tilde{a}$} of \mbox{$a= \| \beta \|^2 = \frac{1}{N} \sum_x F(x)=m$} 
and thus $\tilde{m}$ is directly an estimate of $m$. 
The error $|\tilde{m} - m|$ is at most
\begin{equation}\label{eq}
2\pi  \frac{\sqrt{m(1-m)}}{t}+\frac{ \pi^{2}}{t^{2}} \in O\left(\frac{1}{t} \right)
\end{equation}
with probability at least $8/\pi^2$.
This requires $4t$ evaluations of $F$ because each of the $t$ calls on $A$
and on $A^{-1}$ requires $2$ evaluations of~$F$.
\end{proof}
This theorem states that the error goes down asymptotically linearly with the number of evaluation 
of $F$\@. This is optimal according to Nayak and Wu~\cite{{nayakwu}}, who have proven that in the
general case in which we have no \emph{a priori} knowledge of the possible distribution of outputs,  
an additive error of $\epsilon$ requires an amount of work proportional to
$1/\epsilon$ in the worst case when the function is given as a black box. 
Note that 
for \mbox{$t < (1+\sqrt{2})\pi/\sqrt{m}$}, when our bound on the error exceeds the targeted mean
(which is rather bad), the error goes down quadratically (which is good).

We now present a variation on an algorithm of Aharonov~\cite{aharonov} and analyse its characteristics.  This algorithm is also based on amplitude estimation, but it relies on the fact that points in real interval \mbox{$[0,1]$} can be represented 
in binary as $\ell$-bit strings, where $\ell$ is the precision with which we wish to consider the output of black-box function~$F$.
The~\mbox{algorithm} estimates the number of $1$s in each binary position. 
The difference \mbox{between} our algorithm (\textsf{mean2}) and Aharonov's original algorithm is that we make sure that the estimates of the counts in every bit position are all simultaneously within the desired error bound.
For each $i$ between $1$ and~$\ell$,
let $F_i(x)$ represent the $i^\mathrm{th}$ bit of the binary expansion of~$F(x)$,
so that $F(x) =\sum_i  F_i(x) 2^{-i}$,
with the obvious case $F_i(x)=1$ for all $i$ when $F(x)=1$.

\begin{algorithm}
	\caption{\textsf{mean2}$(F,N)$} \label{mean_algo_2}
\begin{algorithmic}\raggedright
	\FOR{$i=1$ to $\ell$}
	\STATE $\tilde{m_i}=\textsf{majority}(\textsf{count}(F_i(x), 5\pi\sqrt{N}\,), n=\lceil\frac{3}{2}\ell\,\rceil)$\\
	\ENDFOR\\
	\textbf{return} $\tilde{m}=\frac{1}{N}\sum_{i=1}^{\ell} \tilde{m_i} 2^{-i}$
\end{algorithmic}
\end{algorithm}

\begin{theorem}
Given a black-box function
$F:\{0,\ldots,N-1\} \rightarrow [0,1]$
where the output of $F$ has $\ell$ bits of precision,
algorithm \textsf{mean2} outputs an estimate $\tilde{m}$ such that $|\tilde{m} - m| \leqslant \frac{1}{N} \sum_i \sqrt{m_i} \, 2^{-i}$,
where \mbox{$m_i=\sum_x F_i(x)$},
with probability at least~$\pbfrac23$.
The~\mbox{algorithm} requires $O(\sqrt{N}\,\ell \log\ell)$ evaluations of $F$.
\end{theorem}
\begin{proof}
	The proof is a straightforward corollary of Theorems \ref{thm:counting} and \ref{thm:majority}. Using \textsf{count} on each column with $s = 5 \pi\sqrt{N}$ yields an error of
	\begin{equation*}
		|m_i - \hat{m}_i| \leqslant {\textstyle \frac{2}{5}} \sqrt{m_i} + {\textstyle \frac{1}{25}}
	\end{equation*}
	with probability at least~$8/\pi^2$, and hence with probability at least~$\pbfrac23$,
	where $\hat{m}_i$ denotes $\textsf{count}(F_i(x), 5 \pi\sqrt{N}\,)$.
	Using \mbox{\textsf{majority}} with $n = \lceil\frac{3}{2}\ell\rceil$ on this, we obtain
	\begin{equation*}
		|m_i - \tilde{m}_i| \leqslant {\textstyle \frac{4}{5}} \sqrt{m_i} + {\textstyle \frac{2}{25}}
	\end{equation*}
  with probability at least $1-\frac{2}{3\ell}$. When $m_i \geqslant 1$, this is bounded by $\sqrt{m_i}$. Furthermore, \textsf{count} makes no error when~\mbox{$m_i = 0$}. Hence, the error in each column is bounded by $\sqrt{m_i}$ with probability at least $1-\frac{2}{3\ell}$. By the union bound, all of our estimates for the columns are simultaneously within the above error bounds with probability at least~$\pbfrac23$, and the error bound on our final estimate is $|\tilde{m} - m| \leqslant \frac{1}{N} \sum \sqrt{m_i} \, 2^{-i}$.
It~is straightforward to count the number of evaluations of~$F$ from Theorems \ref{thm:counting} and \ref{thm:majority}.
\end{proof}

The choice of which among algorithms \textsf{mean1} or \textsf{mean2} is more appropriate \mbox{depends} on the particular characteristics of the input function. 
For example, one can consider the situation in which 
$F(x)=2^{-N}$ for all~$x$, hence the mean \mbox{$m=2^{-N}$} as well.
In~this case, if we choose $t= c N^{3/2} \lg N$ in \textsf{mean1} and $\ell=N$ in \textsf{mean2},
where constant $c$ is chosen so that both algorithms call function $F$ the same number of times,
the first algorithm is in the regime \mbox{$t \ll (1+\sqrt{2})\pi/\sqrt{m}$} where it performs
badly because the error on the estimated mean is expected to be much larger than the mean itself.
On~the other hand, the expected error produced by the second algorithm is bounded by $m/\sqrt{N}$,
which is much \emph{smaller} than the targeted mean.
At~the other end of the spectrum,
if $F(x)=\pbfrac12$ for all~$x$, hence the mean \mbox{$m=\pbfrac12$} as well,
and if \mbox{$t \gg 2\pi\sqrt{N}$}, then the error produced by \textsf{mean1} is much smaller than~$m/\sqrt{N}$
according to Equation~(\ref{eq}).
With the same parameters, the error produced by \textsf{mean2}, which is again bounded by~$m/\sqrt{N}$,
is strictly unaffected by the choice of~$\ell$, so that the second algorithm can work arbitrarily harder than
the first, yet produce a less precise estimate of the mean.

\section{Approximate Median Algorithm}
\label{approximate_median}

Let \textsf{dist}$:\{0,\ldots,N-1\}\times\{0,\ldots,N-1\}\rightarrow [0,1]$ be an arbitrary black-box distance function. 

\begin{defi}[Median]
The \emph{median} is the point within an ensemble of points whose average distance to the other points is minimum.
\end{defi}

Formally, the median of a set of points \mbox{$Q=\{0,\ldots,N-1\}$}
is
\begin{equation*}
\text{median}(Q) = \argmin_{z \in Q} \sum_{j=0}^{N-1} \textsf{dist}(z,j).
\end{equation*}
The median can be found classically by going through each point $z \in Q$, computing the average
distance from $z$ to all the other points in $Q$, and then taking the minimum (ties are broken arbitrarily).  This process requires a time of~$O(N^2)$.
In~the general case, in which there are no restrictions on the distance function used and no structure among the ensemble of points that can be exploited, no technique can be more efficient than this na\"{\i}ve algorithm.
Indeed, consider the case in which all the points are at the same distance from each other, except for two points that are closer than the rest of the points. These two points are the medians of this ensemble. In this case, classically we would need to query the oracle for the distances between each and every
pair of points before we can identify one of the two medians. (We expect to discover this special pair after querying about half the pairs on the average but we cannot know that there isn't some \emph{other} even closer pair until \emph{all} the pairs have been queried.) This results in a lower bound of $\Omega(N^2)$ calls to the oracle.

In Algorithm 3, \textsf{mean} stands for either one of the two algorithms given in the previous section (in~case \textsf{mean1}
is used, parameter $t$ must be added) but it is repeated $O(\log N)$ times in order to get all the means within the desired error bound
with a constant probability via our \textsf{majority} algorithm (Theorem~\ref{thm:majority}).
Here, $d_i = \frac{1}{N} \sum_j \textsf{dist}(i,j)$ and $d_{\min}=d_k$ for any $k$ such that $d_k \leqslant d_i$ for all~$i$.

\begin{algorithm}
	\caption{\textsf{median}(\textsf{dist})} \label{algo:median}
\begin{algorithmic}\raggedright
\STATE For each $i$, define function $F_i(x)=\textsf{dist}(i,x)$
\STATE For each $i$, define $\tilde{d_i}=\textsf{majority}(\textsf{mean}(F_i,N), n=N^2)$
\STATE \textbf{return} \textsf{minimum}$(\tilde{d_i})$
\end{algorithmic}
\end{algorithm}

\begin{theorem}\label{thm:median-with-mean1}
For any black-box distance function
\[\mathsf{dist}:\{0,\ldots,N-1\}\times\{0,\ldots,N-1\}\rightarrow [0,1] \, , \]
when \textsf{mean1} is used with parameter~$t$, algorithm \textsf{median} outputs an index $j$ such that \mbox{$|d_j - d_{\min}| \in O(1/t)$} with probability at least $\pbfrac23$. The algorithm requires $O( t \sqrt{N}\log N)$ evaluations of \textsf{dist}.
\end{theorem}

\begin{proof}
This result is obtained by a straightforward combination of Theorems \ref{thm:minimum}, \ref{thm:majority} and~\ref{thm:mean1}. 
The procedure \textsf{majority} is used with parameter \mbox{$n=N^2$} to ensure that all the $d_i$'s computed by the algorithm (in superposition) are simultaneously within the bound given by Theorem~\ref{thm:mean1}, except with probability~$o(1)$. Note that with parameter \mbox{$n=N^2$}, the number of repetitions is still in $O(\log N)$.  The success probability of the algorithm follows from the fact that \mbox{$\frac34(1-o(1))>\pbfrac23$}. In this case, the error is in $O(1/t)$ and the number of evaluations of \textsf{dist} is in $O(t \sqrt{N} \log N)$.
\end{proof}

By replacing \textsf{mean1} by  \textsf{mean2} in the \textsf{median} algorithm we obtain the following theorem.

\begin{theorem}\label{thm:median-with-mean2}
For any black-box distance function
\[ \mathsf{dist}:\{0,\ldots,N-1\}\times\{0,\ldots,N-1\}\rightarrow [0,1] \, , \]
when \textsf{mean2} is used, algorithm \textsf{median} outputs an index $j$ such that 
\[ |d_j - d_{\min}| \leqslant \frac{1}{N}\sum_{i=1}^\ell \sqrt{m_i} \, 2^{-i} \]
with probability at least~$\pbfrac23$.
(See algorithm \textsf{mean2} for a definition of $m_i$ and~$\ell$.) 
The~algo\-rithm requires $O(N\log N)$ evaluations of \textsf{dist}.
\end{theorem}

\section{Conclusion}\label{conclu}

We have described two quantum algorithms to approximate the mean and their applications to 
approximate the median of a set of points over an arbitrary distance function given by a black box. We~leave open for future work an in-depth study on how the different behaviour of the two algorithms impact the quality of the median they return. For instance, we know that the behaviour of both algorithms for the mean depends on the distribution of data points and the distances between points,
but we still have to inves\-ti\-gate more precisely the exact context where it matters. Of course, understanding the behaviour of the algorithms in different contexts is important, but a more interesting question is to tailor the algorithm to obtain better results 
on different data distributions of interest.

\end{document}